%% file: main.tex
\documentclass[11pt]{llncs}
\usepackage{amsmath,amssymb}
\usepackage{tabls}
\usepackage{graphicx}
\usepackage{wrapfig}
\usepackage{array}
\usepackage[algoruled,linesnumbered,algo2e,vlined]{algorithm2e}
\usepackage{url}
\usepackage{amssymb}
\usepackage{multirow}
\usepackage{multicol}
\usepackage{fullpage}
\newcommand{\Occ}{\mathit{Occ}}

\newcommand{\derive}{\mathit{val}}
\newcommand{\deriveInt}{\mathit{itv}}
\newcommand{\VarOcc}{\mathit{vOcc}}
\newcommand{\suffix}{\mathit{suf}}
\newcommand{\prefix}{\mathit{pre}}
\newcommand{\lmq}{\mathit{lm}_q}
\newcommand{\rmq}{\mathit{rm}_q}
\newcommand{\trie}{\Upsilon}

\newcommand{\nodelabel}{\mathit{label}}
\newcommand{\variable}[1]{X_{\langle #1\rangle}}
\newcommand{\kushi}{\xi_\mathcal{T}}

\newlength\savedwidth

\newcommand{\ssa}{\textrm{SSA}}
\newcommand{\stsa}{\textrm{STSA}}
\newcommand{\nsa}{\textrm{NSA}}

\author{
  Keisuke Goto
  \and
  Hideo Bannai
  \and
  Shunsuke Inenaga
  \and
  Masayuki Takeda
}
\institute{
  Department of Informatics, Kyushu University\\
  \email{\{keisuke.gotou,bannai,inenaga,takeda\}@inf.kyushu-u.ac.jp}\\
}
\title{
  Speeding-up $q$-gram mining on grammar-based compressed texts
}
\date{}

\begin{document}
\maketitle
\begin{abstract}
  We present an efficient algorithm for calculating
  $q$-gram frequencies on strings represented in compressed form,
  namely, as a straight line program (SLP).
  Given an SLP $\mathcal{T}$ of size $n$ that represents string $T$,
  the algorithm computes the occurrence frequencies of
  {\em all} $q$-grams in $T$, by reducing the problem to the weighted
  $q$-gram frequencies problem on a trie-like structure of size
  $m = |T|-\mathit{dup}(q,\mathcal{T})$,
  where $\mathit{dup}(q,\mathcal{T})$ is a quantity
  that represents the amount of redundancy that the SLP captures with respect to
  $q$-grams.
  The reduced problem can be solved in linear time.
  Since $m = O(qn)$, the running time of our algorithm is
  $O(\min\{|T|-\mathit{dup}(q,\mathcal{T}),qn\})$,
  improving our previous $O(qn)$ algorithm when
  $q = \Omega(|T|/n)$.  
\end{abstract}
\input{introduction}
\input{preliminaries}
\input{algorithm}
\input{experiments}
\bibliographystyle{splncs03}
\bibliography{ref}
\end{document}

%% file: introduction.tex
\section{Introduction}
Many large string data sets are usually first compressed and
stored, while they are decompressed afterwards in order to be used and
analyzed.
Compressed string processing (CSP) is an approach that has been
gaining attention in the string processing community. 
Assuming that the input is given in compressed form,
the aim is to develop methods where the string is processed or analyzed without
explicitly decompressing the entire string, 
leading to algorithms with time and space complexities
that depend on the compressed size rather than the 
whole uncompressed size.
Since compression algorithms inherently capture regularities
of the original string, clever CSP algorithms
can be
theoretically~\cite{NJC97,crochemore03:_subquad_sequen_align_algor_unres_scorin_matric,hermelin09:_unified_algor_accel_edit_distan,gawrychowski11:_LZ_comp_str_fast_}, and
even
practically~\cite{shibata00:_speed_up_patter_match_text_compr,goto11:_fast_minin_slp_compr_strin},
faster than algorithms which process the uncompressed string.

In this paper, we assume that the input string
is represented as a Straight Line Program (SLP),
which is a context free grammar in Chomsky normal form
that derives a single string.
SLPs are a useful tool when considering CSP algorithms, since
it is known that outputs of various
grammar based compression algorithms~\cite{SEQUITUR,LarssonDCC99},
as well as dictionary compression algorithms~\cite{LZ78,LZW,LZ77,LZSS}
can be modeled efficiently by SLPs~\cite{rytter03:_applic_lempel_ziv}.
We consider the $q$-gram frequencies problem on compressed text
represented as SLPs.
$q$-gram frequencies have profound applications in the field of string
mining and classification.
The problem was first considered for the CSP setting
in~\cite{inenaga09:_findin_charac_subst_compr_texts},
where an $O(|\Sigma|^2n^2)$-time $O(n^2)$-space algorithm
for finding the {\em most frequent} $2$-gram from an
SLP of size $n$ representing text $T$ over alphabet $\Sigma$ was presented.
In~\cite{claudear:_self_index_gramm_based_compr}, it is claimed
that the most frequent $2$-gram can be found in $O(|\Sigma|^2n\log n)$-time
and $O(n\log|T|)$-space, if the SLP is pre-processed and a self-index is built. 
A much simpler and efficient $O(qn)$
time and space algorithm for general $q \geq 2$ was recently
developed~\cite{goto11:_fast_minin_slp_compr_strin}.

Remarkably, computational experiments on various data sets
showed that the $O(qn)$ algorithm is actually faster than
calculations on uncompressed strings, when $q$ is
small~\cite{goto11:_fast_minin_slp_compr_strin}.
However, the algorithm
slows down considerably
compared to the uncompressed approach when $q$ increases.
This is because the algorithm reduces the 
$q$-gram frequencies problem on an SLP of size $n$, 
to the weighted  $q$-gram frequencies problem on a weighted string of
size at most $2(q-1)n$.
As $q$ increases, the length of the string becomes longer than the
uncompressed string $T$.
Theoretically $q$ can be as large as $O(|T|)$,
hence in such a case the algorithm requires $O(|T|n)$ time,
which is worse than a trivial $O(|T|)$ solution that first decompresses the given SLP and 
runs a linear time algorithm for $q$-gram frequencies computation on $T$.

In this paper, we solve this problem, and improve the previous $O(qn)$ algorithm
both theoretically and practically.
We introduce a $q$-gram neighbor relation on SLP variables,
in order to reduce the redundancy
in the partial decompression of the string which is performed in the
previous algorithm.
Based on this idea, we are able to convert the problem to a
weighted $q$-gram frequencies problem
on a weighted trie, whose size is at most
$|T|-\mathit{dup}(q,\mathcal{T})$.
Here, $\mathit{dup}(q,\mathcal{T})$ is a quantity that represents
the amount of redundancy that the SLP captures with respect to
$q$-grams. Since the size of the trie is also bounded by $O(qn)$,
the time complexity of our new algorithm is
$O(\min\{qn,|T|-\mathit{dup}(q,\mathcal{T})\})$,
improving on our previous $O(qn)$ algorithm when
$q = \Omega(|T|/n)$.
Preliminary computational experiments show that our new approach achieves
a practical speed up as well, for all values of $q$.

%% file: preliminaries.tex
\section{Preliminaries}
\subsection{Intervals, Strings, and Occurrences}
For integers $i \leq j$, let $[i:j]$ denote the interval of integers $\{i,\ldots, j\}$.
For an interval $[i:j]$ and integer $q > 0$,
let $\prefix([i:j],q)$ and $\suffix([i:j],q)$
represent respectively, the length-$q$ prefix and suffix interval,
that is,
$\prefix([i:j],q) = [i:\min(i+q-1,j)]$ and
$\suffix([i:j],q) = [\max(i,j-q+1):j]$.

Let $\Sigma$ be a finite {\em alphabet}.
An element of $\Sigma^*$ is called a {\em string}.
For any integer $q > 0$, an element of $\Sigma^q$ is called a \emph{$q$-gram}.
The length of a string $T$ is denoted by $|T|$. 
The empty string $\varepsilon$ is a string of length 0,
namely, $|\varepsilon| = 0$.
For a string $T = XYZ$, $X$, $Y$ and $Z$ are called
a \emph{prefix}, \emph{substring}, and \emph{suffix} of $T$, respectively.
The $i$-th character of a string $T$ is denoted by $T[i]$, where $1 \leq i \leq |T|$.
For a string $T$ and interval $[i:j] (1 \leq i \leq j \leq |T|)$, let $T([i:j])$
denote the  substring of $T$ that begins at position $i$ and ends at
position $j$.
For convenience, let $T([i:j]) = \varepsilon$ if $j < i$.
For a string $T$ and integer $q \geq 0$,
let $\prefix(T,q)$ and $\suffix(T,q)$
represent respectively, the length-$q$ prefix and suffix of $T$,
that is,
$\prefix(T,q) =T(\prefix([1:|T|],q))$ and $\suffix(T,q) = T(\suffix([1:|T|],q))$.

For any strings $T$ and $P$,
let $\Occ(T,P)$ be the set of occurrences of $P$ in $T$, i.e.,
$\Occ(T,P) = \{k > 0 \mid T[k:k+|P|-1] = P\}$.
The number of elements $|\Occ(T,P)|$ is called
the \emph{occurrence frequency} of $P$ in $T$.

\subsection{Straight Line Programs}

\begin{wrapfigure}[11]{r}{0.5\textwidth}
  \vspace{-2.45cm}
  \centerline{\includegraphics[width=0.45\textwidth]{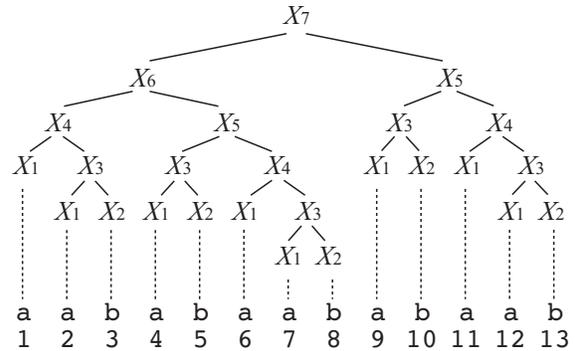}}
  \caption{
    The derivation tree of
    SLP $\mathcal T = \{ X_1 \rightarrow \mathtt{a}$, $X_2 \rightarrow \mathtt{b}$, $X_3 \rightarrow X_1X_2$,
    $X_4 \rightarrow X_1X_3$, $X_5 \rightarrow X_3X_4$, $X_6
    \rightarrow X_4X_5$, $X_7 \rightarrow X_6X_5 \}$,
    representing string $T = \derive(X_7) = \mathtt{aababaababaab}$.
  }
  \label{fig:SLP}
\end{wrapfigure}

A {\em straight line program} ({\em SLP}) is a set of assignments 
$\mathcal T = \{ X_1 \rightarrow expr_1, X_2 \rightarrow expr_2, \ldots, X_n \rightarrow expr_n\}$,
where each $X_i$ is a variable and each $expr_i$ is an expression, where
$expr_i = a$ ($a\in\Sigma$), or $expr_i = X_{\ell(i)} X_{r(i)}$~($i > \ell(i),r(i)$).
It is essentially a context free grammar in the Chomsky normal form, that derives a single string.
Let $\derive(X_i)$ represent the string derived from variable $X_i$.
To ease notation, we sometimes associate $\derive(X_i)$ with $X_i$ and
denote $|\derive(X_i)|$ as $|X_i|$,
and $\derive(X_i)([u:v])$ as $X_i([u:v])$ for any interval $[u:v]$.
An SLP $\mathcal{T}$ {\em represents} the string $T = \derive(X_n)$.
The \emph{size} of the program $\mathcal T$ is the number $n$ of
assignments in $\mathcal T$. 
Note that $|T|$ can be as large as $\Theta(2^n)$. However, we assume
as in various previous work on SLP, 
that the computer word size is at least $\log |T|$, and hence,
values representing lengths and positions of $T$
in our algorithms can be manipulated in constant time.

The derivation tree of SLP $\mathcal{T}$ is a labeled
ordered binary tree where each internal node is labeled with a
non-terminal variable in $\{X_1,\ldots,X_n\}$, and each leaf is labeled with a terminal character in $\Sigma$.
The root node has label $X_n$.
Let $\mathcal{V}$ denote the set of internal nodes in
the derivation tree.
For any internal node $v\in\mathcal{V}$, 
let $\langle v\rangle$ denote the index of its label
$\variable{v}$.
Node $v$ has a single child which is a leaf labeled with $c$
when $(\variable{v} \rightarrow c) \in \mathcal{T}$ for some $c\in\Sigma$,
or
$v$ has a left-child and right-child respectively denoted $\ell(v)$ and $r(v)$,
when
$(\variable{v}\rightarrow \variable{\ell(v)}\variable{r(v)}) \in \mathcal{T}$.
Each node $v$ of the tree derives $\derive(\variable{v})$,
a substring of $T$,
whose corresponding interval $\deriveInt(v)$,
with $T(\deriveInt(v)) = \derive(\variable{v})$,
can be defined recursively as follows.
If $v$ is the root node, then $\deriveInt(v) = [1:|T|]$.
Otherwise, if $(\variable{v}\rightarrow
\variable{\ell(v)}\variable{r(v)})\in\mathcal{T}$,
then,
$\deriveInt(\ell(v)) = [b_v:b_v+|\variable{\ell(v)}|-1]$
and
$\deriveInt(r(v)) = [b_v+|\variable{\ell(v)}|:e_v]$,
where $[b_v:e_v] = \deriveInt(v)$.
Let $\VarOcc(X_i)$ denote the number of times a variable $X_i$ occurs
in the derivation tree, i.e., 
$\VarOcc(X_i) = |\{ v \mid \variable{v}=X_i\}|$.
We assume that any variable $X_i$ is used at least once,
that is $\VarOcc(X_i) > 0$.

For any interval $[b:e]$ of $T (1\leq b \leq e \leq |T|)$, 
let $\kushi(b,e)$ denote the deepest node $v$ in the derivation tree,
which derives an interval containing $[b:e]$, that is,
$\deriveInt(v)\supseteq [b:e]$,
and no proper descendant of $v$
satisfies this condition.
We say that node $v$ {\em stabs} interval $[b:e]$,
and $\variable{v}$ is called the variable that stabs the interval.
If $b = e$, we have that
$(\variable{v} \rightarrow c) \in \mathcal{T}$ for some $c\in\Sigma$,
and $\deriveInt(v) = b = e$.
If $b < e$, then
we have $(\variable{v} \rightarrow
\variable{\ell(v)}\variable{r(v)})\in\mathcal{T}$,
$b\in \deriveInt(\ell(v))$, and  $e\in\deriveInt(r(v))$.
When it is not confusing, we will sometimes 
use $\kushi(b,e)$ to denote the variable $\variable{\kushi(b,e)}$.

SLPs can be efficiently pre-processed to hold various information.
$|X_i|$ and $\VarOcc(X_i)$ can be computed for all variables $X_i
(1\leq i\leq n)$ in a total of $O(n)$ time by a simple dynamic
programming algorithm.
Also, the following Lemma is useful for partial decompression of
a prefix of a variable.
\begin{lemma}[\cite{gasieniec05:_real_time_traver_gramm_based_compr_files}]
  \label{label:prefix_decompression}
  Given an SLP $\mathcal{T} = \{ X_i \rightarrow \mathit{expr}_i \}_{i=1}^n$,
  it is possible to pre-process $\mathcal{T}$ in $O(n)$ time and
  space, so that for any variable $X_i$ and $1 \leq j \leq |X_i|$,
  ${X_i}([1:j])$ can be computed in $O(j)$ time.
\end{lemma}

The formal statement of the problem we solve is:
\begin{problem}[$q$-gram frequencies on SLP]
  \label{problem:SLPqgramfreq}
  Given integer $q\geq 1$ and an SLP $\mathcal{T}$ of size $n$ that represents string $T$,
  output $(i, |\Occ(T,P)|)$ for all $P\in\Sigma^q$ where
  $\Occ(T,P)\neq\emptyset$, and some $i\in\Occ(T,P)$.
\end{problem}
Since the problem is very simple for $q = 1$, 
we shall only consider the case for $q\geq 2$ for the rest of the paper.
Note that although the number of distinct $q$-grams in $T$ is bounded by $O(qn)$,
we would require an extra multiplicative $O(q)$ factor for the output if we output
each $q$-gram explicitly as a string.
In our algorithms to follow, we compute 
a compact, $O(qn)$-size representation of the output,
from which each $q$-gram can be easily obtained in $O(q)$ time.

%% file: algorithm.tex
\section{$O(qn)$ Algorithm~\cite{goto11:_fast_minin_slp_compr_strin}}
\label{section:qn}
In this section, we briefly describe the $O(qn)$ algorithm
presented in~\cite{goto11:_fast_minin_slp_compr_strin}.
The idea is to count occurrences of $q$-grams with respect to
the variable that stabs its occurrence.
The algorithm reduces Problem~\ref{problem:SLPqgramfreq}
to calculating the frequencies of all $q$-grams in a weighted set of strings,
whose total length is $O(qn)$.
Lemma~\ref{lemma:qn_key}
shows the key idea of the algorithm.
\begin{lemma}
  \label{lemma:qn_key}
  For any SLP $\mathcal{T} = \{ X_i \rightarrow \mathit{expr}_i \}_{i=1}^n$
  that represents string $T$, integer $q \geq 2$, and $P \in \Sigma^q$,
  $|\Occ(T,P)| = \sum_{i=1}^n\VarOcc(X_i)\cdot |\Occ(t_i, P)|$,
  where
  $t_i = \suffix(\derive(X_{\ell(i)}),q-1)\prefix(\derive(X_{r(i)}),q-1)$.
\end{lemma}
\begin{proof}
  For any $q \geq 2$, $v$ stabs the interval $[u:u+q-1]$
  if and only if
  $[u:u+q-1]\subseteq
  [s_v:f_v] = 
  \suffix(\deriveInt(\ell(v)),q-1)\cup\prefix(\deriveInt(r(v)),q-1)$. 
  (See Fig.~\ref{fig:SLP-kgram}.)
  Also, since an occurrence of $X_i$ in the derivation tree always
  derives the same string $\derive(X_i)$, 
  $t_i = T([s_v:f_v])$ for any node $v$ such that $\variable{v} = X_i$.
  Therefore,
  \begin{eqnarray*}
    \lefteqn{|\Occ(T,P)| = \big|\{ u>0 \mid  T([u:u+q-1]) = P\}\big|}\\
    & = & \sum_{v\in \mathcal{V}} \big| \{ u>0 \mid \xi_{\mathcal{T}}(u,u+q-1)=v, j=u-s_v+1, \variable{v}([j:j+q-1]) = P \}\big| \\
    & = & \sum_{i=1}^n \sum_{v\in \mathcal{V}: \variable{v}=X_i} \big|\{
    u>0 \mid  \xi_{\mathcal{T}}(u,u+q-1)=v,j=u-s_v+1, \variable{v}([j:j+q-1]) = P\}\big|\\
    & = & \sum_{i=1}^n \sum_{v\in \mathcal{V}: \variable{v}=X_i} \Occ(T([s_v:f_v]),P)
    = \sum_{i=1}^n \VarOcc(X_i)\cdot \Occ(t_i,P).\\
  \end{eqnarray*}
  \qed
\end{proof}

\begin{wrapfigure}[15]{r}{0.4\textwidth}
   \vspace{-1cm}
  \begin{center} 
    \includegraphics[width=0.4\textwidth]{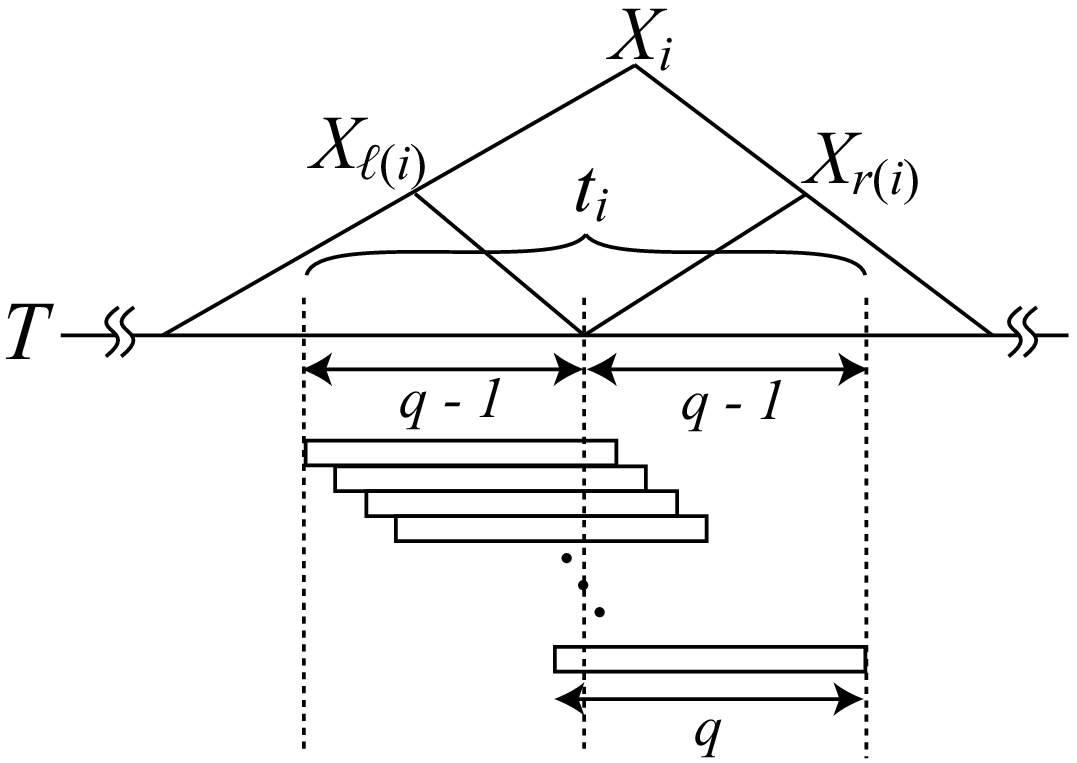}
  \end{center}
  \caption{
    Length-$q$ intervals where
    $\variable{\kushi(u,u+q-1)} = X_i$,
    and $(X_i\rightarrow X_{\ell(i)} X_{r(i)}) \in \mathcal{T}$.
  }
  \label{fig:SLP-kgram}
\end{wrapfigure}

From Lemma~\ref{lemma:qn_key}, we have that occurrence frequencies in
$T$ are equivalent to occurrence frequencies in $t_i$ weighted by $\VarOcc(X_i)$.
Therefore, the $q$-gram frequencies problem can be regarded as
obtaining the {\em weighted} frequencies of all $q$-grams in the set
of strings $\{t_1,\ldots,t_n\}$,
where each occurrence of a $q$-gram in $t_i$ is weighted by $\VarOcc(X_i)$.
This can be further reduced to a weighted $q$-gram frequency
problem for a single string $z$, where each position of $z$
holds a weight associated with the $q$-gram that starts at that position.
String $z$ is constructed by concatenating all $t_i$'s with length at
least $q$.
The weights of positions corresponding to the first $|t_i| - (q-1)$
characters of $t_i$ will be $\VarOcc(X_i)$, while the last $(q-1)$
positions will be $0$ so that superfluous $q$-grams generated by the
concatenation are not counted.
The remaining is a simple linear time algorithm
using suffix and lcp arrays on the weighted string,
thus solving the problem in $O(qn)$ time and space.

\section{New Algorithm}
\label{section:new_algorithm}
We now describe our new algorithm which solves the $q$-gram
frequencies problem on SLPs.
The new algorithm basically follows the previous $O(qn)$ algorithm,
but is an elegant refinement.
The reduction for the previous $O(qn)$ algorithm
leads to a fairly large amount of redundantly decompressed regions of
the text as $q$ increases.
This is due to the fact that
the $t_i$'s are considered independently for each variable $X_i$,
while {\em neighboring} $q$-grams that are stabbed by different
variables actually share $q-1$ characters.
The key idea of our new algorithm is to exploit this redundancy.
(See Fig.~\ref{fig:qgramneighbor}.)
In what follows, we introduce the concept of $q$-gram neighbors,
and reduce the $q$-gram frequencies problem on SLP to
a weighted $q$-gram frequencies problem on a weighted tree.
\subsection{$q$-gram Neighbor Graph}
We say that
$X_j$ is a {\em right $q$-gram neighbor} of $X_i$ $(i \neq j)$,
or equivalently,
$X_i$ is a {\em left $q$-gram neighbor} of $X_j$,
if for some integer $u \in [1:|T|-q]$,
$\variable{\kushi(u,u+q-1)} = X_i$ and
$\variable{\kushi(u+1,u+q)} = X_j$.
Notice that $|X_i|$ and $|X_j|$ are both at least $q$
if $X_i$ and $X_j$ are right or left $q$-gram neighbors of each other.

\begin{figure}[t]
  \centerline{ 
    \includegraphics[width=0.45\textwidth]{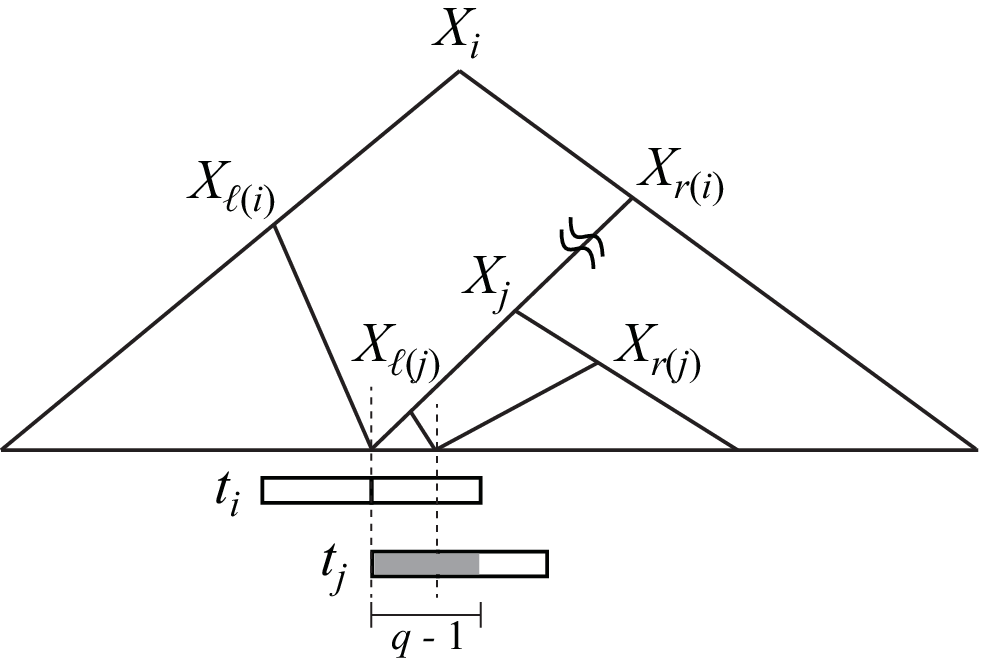}
    \hfill
    \includegraphics[width=0.45\textwidth]{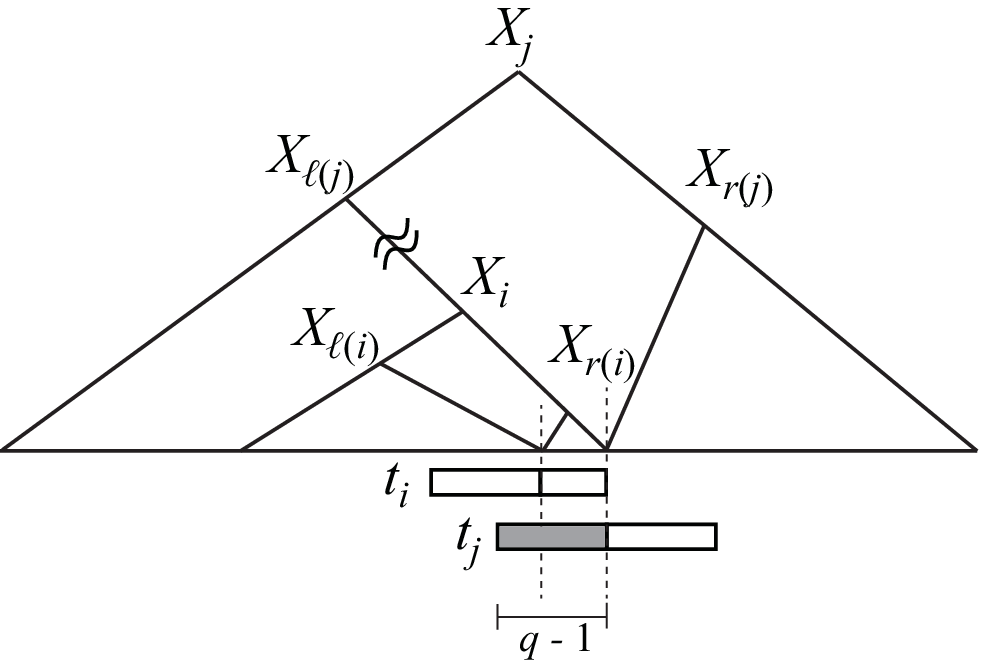}
  }
  \caption{$q$-gram neighbors and redundancies.
    (Left) $X_j$ is a right $q$-gram neighbor of
    $X_i$, and $X_i$ is {\em a} left $q$-gram neighbor of $X_j$.
    Note that
    the right $q$-gram neighbor of $X_i$ is uniquely determined since
    $|X_{r(i)}| \geq q$ and it must be a descendant on the left most path
    rooted at $X_{r(i)}$,
    However, $X_{j}$ may have other left $q$-gram neighbors, since
    $|X_{\ell(j)}| < q$, and they must be ancestors of $X_j$.
    $t_i$ (resp. $t_j$) represents the string corresponding to
    the union of intervals $[u:u+q-1]$ where $\variable{\kushi(u,u+q-1)} = X_i$ (resp. $\variable{\kushi(u,u+q-1)}=X_j$).
    The shaded region depicts the string which is redundantly
    decompressed, if both $t_i$ and $t_j$ are considered independently.
    (Right) Shows the reverse case, when $|X_{r(i)}| < q$.
  }
  \label{fig:qgramneighbor}
\end{figure}

\begin{definition}
  For $q\geq 2$, the right $q$-gram neighbor graph of SLP
  $\mathcal{T} = \{ X_i \rightarrow expr_i \}_{i=1}^n$ is the directed graph
  $G_q = (V,E_r)$, where
  \begin{eqnarray*}
    V&=&\{ X_i \mid i \in \{1,\ldots, n\}, |X_i| \geq q \}\\
    E_r &=& \{ (X_i,X_j) \mid X_j \mbox{ is a right $q$-gram neighbor of
      $X_i$ } \}
  \end{eqnarray*}
\end{definition}

Note that there can be multiple right $q$-gram neighbors for a given variable.
However, 
the total number of edges in the neighbor graph 
is bounded by $2n$, as will be shown below.

\begin{lemma}
  \label{lemma:unique_neighbors}
  Let $X_j$ be a right $q$-gram neighbor of $X_i$.
  If, $|X_{r(i)}| \geq q$, then $X_j$ is the label of the deepest
  variable on the left-most path of the derivation tree rooted at
  a node labeled $X_{r(i)}$ whose length is at least $q$.
  Otherwise, if $|X_{r(i)}| < q$, then
  $X_i$ is the label of the deepest variable on the right-most path rooted at
  a node labeled $X_{\ell(j)}$
  whose length is at least $q$.
\end{lemma}
\begin{proof}
  Suppose $|X_{r(i)}| \geq q$.
  Let $u$ be a position, where
  $\variable{\kushi(u,u+q-1)} = X_i$ and $\variable{\kushi(u+1,u+q)} = X_j$.
  Then, since the interval $[u+1:u+q]$ is a prefix of
  $\deriveInt(X_{r(i)})$,
  $X_j$ must be on the left most path rooted at $X_{r(i)}$. 
  Since $X_j = \variable{\kushi(u+1,u+q)}$,
  the lemma follows from the definition of $\kushi$.
  The case for $|X_{r(i)}| < q$ is symmetrical
  and can be shown similarly.
  \qed
\end{proof}

\begin{lemma}
  For an arbitrary SLP
  $\mathcal{T} = \{ X_i \rightarrow expr_i \}_{i=1}^n$ 
  and integer $q\geq 2$,
  the number of edges in the right $q$-gram neighbor graph $G_q$ of $\mathcal{T}$ is at most $2n$.
\end{lemma}
\begin{proof}
  Suppose $X_j$ is a right $q$-gram neighbor of $X_i$.
  From Lemma~\ref{lemma:unique_neighbors}, we have that
  if $|X_{r(i)}| \geq q$, the right $q$-gram neighbor of $X_i$ is uniquely determined
  and that $|X_{\ell(j)}| < q$.
  Similarly, if $|X_{r(i)}| < q$, $|X_{\ell(j)}|\geq q$ and the left $q$-gram neighbor of $X_j$
  is uniquely $X_i$.  
  Therefore,
  \begin{eqnarray*}
    &&\sum_{i=1}^n |\{ (X_i,X_j) \in E_r \mid  |X_{r(i)}| \geq q \}|
    + \sum_{i=1}^n |\{ (X_i,X_j) \in E_r \mid |X_{r(i)}| < q \}|\\
    &=&
    \sum_{i=1}^n   |\{ (X_i,X_j) \in E_r \mid  |X_{r(i)}| \geq q \}|
    + \sum_{i=1}^n |\{ (X_i,X_j) \in E_r \mid  |X_{\ell(j)}| \geq q \}| \leq 2n.
  \end{eqnarray*}
\qed
\end{proof}

\begin{lemma}
  For an arbitrary SLP $\mathcal{T} = \{ X_i \rightarrow expr_i \}_{i=1}^n$ and
  integer $q\geq 2$,
  the right $q$-gram neighbor graph $G_q$ of $\mathcal{T}$ can be constructed in $O(n)$ time.
\end{lemma}
\begin{proof}
  For any variable $X_i$, let
  $\lmq(X_i)$ and $\rmq(X_i)$ respectively represent the 
  index of the label of the deepest
  node with length at least $q$
  on the left-most and right-most path
  in the derivation tree rooted at $X_i$, or $\mathit{null}$ if $|X_i| < q$.
  These values can be computed for all variables in
  a total of $O(n)$ time based on the following recursion:
  If $(X_i \rightarrow a)\in\mathcal{T}$ for some $a\in\Sigma$, then 
  $\lmq(X_i) = \rmq(X_i) = \mathit{null}$.
  For $(X_i \rightarrow X_{\ell(i)}X_{r(i)})\in\mathcal{T}$,
  \begin{equation*}
    \lmq(X_i) = \begin{cases}
      \mathit{null}    & \mbox{if } |X_i| < q,\\
      i              & \mbox{if } |X_i| \geq q \mbox{ and } |X_{\ell(i)}| < q,\\
      \lmq(X_{\ell(i)}) & \mbox{otherwise. }
    \end{cases}
  \end{equation*}
  $\rmq(X_i)$ can be computed similarly.
  Finally,
  \begin{eqnarray*}
    E_r & = &
    \{ (X_i, X_{\lmq(X_{r(i)})}) \mid \lmq(X_{r(i)}) \neq \mathit{null}, 
    i = 1,\ldots, n \}\\
    &&\cup \{ (X_{\rmq(X_{\ell(i)})}, X_i) \mid \rmq(X_{\ell(i)}) \neq
    \mathit{null},
    i=1,\ldots, n\}.
  \end{eqnarray*}
  \qed
\end{proof}

\begin{lemma}
  \label{lemma:connected}
  Let $G_q = (V, E_r)$ be the right $q$-gram neighbor graph of 
  SLP $\mathcal{T} = \{ X_i = expr_i \}_{i=1}^n$ representing string $T$,
  and let $X_{i_1} = \variable{\kushi(1,q)}$.
  Any variable $X_j \in V (i_1 \neq j)$ is reachable from $X_{i_1}$,
  that is, there exists a directed path from $X_{i_1}$ to $X_j$ in $G_q$.
\end{lemma}
\begin{proof}
Straightforward, since any $q$-gram of $T$ except for the left most
$T([1:q])$ has a $q$-gram on its left.\qed
\end{proof}

\subsection{Weighted $q$-gram Frequencies Over a Trie}
\label{section:weighted_q-gram_frequencies_over_a_trie}
From Lemma~\ref{lemma:connected}, we have that the right $q$-gram
neighbor graph is connected. 
Consider an arbitrary directed spanning tree
rooted at $X_{i_1} = \variable{\kushi(1,q)}$ which can be obtained in linear time by a depth
first traversal on $G_q$ from $X_{i_1}$.
We define the label $\nodelabel(X_i)$ of each node $X_i$ of the
$q$-gram neighbor graph, by \[\nodelabel(X_i) = t_i[q:|t_i|] \]
where $t_i =
\suffix(\derive(X_{\ell(i)}),q-1)\prefix(\derive(X_{r(i)}),q-1)$ as before.
For convenience, let $X_{i_0}$ be a dummy variable
such that
$\nodelabel(X_{i_0}) = T([1:q-1])$, and
$X_{r(i_0)} = X_{i_1}$ (and so $(X_{i_0},X_{i_1})\in E_r$).

\begin{lemma}
  \label{lemma:path}
  Fix a directed spanning tree on the right $q$-gram neighbor graph
  of SLP $\mathcal{T}$, rooted at $X_{i_0}$.
  Consider a directed path $X_{i_0}, \ldots, X_{i_m}$ on the spanning tree.
  The weighted $q$-gram frequencies on the string obtained by
  the concatenation $\nodelabel(X_{i_0}) \nodelabel(X_{i_1}) \cdots \nodelabel(X_{i_m})$,
  where each occurrence of a $q$-gram that ends in a position in
  $\nodelabel(X_{i_j})$ is weighted by
  $\VarOcc(X_{i_j})$,
  is equivalent to the
  weighted $q$-gram frequencies of
  strings $\{t_{i_1}, \ldots t_{i_m}\}$
  where each $q$-gram in $t_{i_j}$ is weighted by $\VarOcc(X_{i_j})$.
\end{lemma}
\begin{proof}
  Proof by induction:
  for $m = 1$,
  we have that $\nodelabel(X_{i_0})\nodelabel(X_{i_1}) = t_{i_1}$.
  All $q$-grams in $t_{i_1}$ end in $t_{i_1}$
  and so are weighted by $\VarOcc(X_{i_1})$.
  When $\nodelabel(X_{i_j})$  is added to 
  $\nodelabel(X_{i_0}) \cdots \nodelabel(X_{i_{j-1}})$,
  $|\nodelabel(X_{i_j})|$ new $q$-grams are formed, which correspond to
  $q$-grams in $t_{i_j}$, i.e. $|t_{i_j}| = q - 1 + |\nodelabel(X_{i_j})|$,
  and $t_{i_j}$ is a suffix of $\nodelabel(X_{i_{j-1}})\nodelabel(X_{i_{j}})$.
  All the new $q$-grams end in $\nodelabel(X_{i_j})$ and are thus
  weighted by $\VarOcc(X_{i_j})$.
  \qed
\end{proof}

From Lemma~\ref{lemma:path}, we can construct a weighted trie $\trie$
based on a directed spanning tree of $G_q$ and $\nodelabel()$, 
where 
the weighted $q$-grams in $\trie$ (represented as length-$q$ paths)
correspond to the 
occurrence frequencies of $q$-grams in $T$.
\footnote{A minor technicality is that a node in $\trie$ may have 
multiple children with the same character label, but this 
does not affect the time complexities of the algorithm.}

\begin{algorithm2e}[t]
\caption{Constructing weighted trie from SLP}
\label{algo:slp2trie}
Construct right $q$-gram neighbor graph $G=(V,E_r)$\;
Calculate $\VarOcc(X_i)$ for $i = 1,\ldots, n$\;
Calculate $|\nodelabel(X_i)|$ for $i = 1,\ldots, n$\;
\lFor{$i = 0,\ldots, n$}{
  $\mathsf{visited}[i] = \mathsf{false}$\;
}
$X_{i_1} = \variable{\kushi(1,q)} = \lmq(X_n)$\;
Define $X_{i_0}$ so that $X_{r(i_0)} = X_{i_1}$ and $|\nodelabel(X_{i_0})| = q-1$\;
$\mathit{root} \leftarrow$ new node\tcp*[l]{root of resulting trie}
\ref{procedure:bdf}($i_0$, $\mathit{root}$)\;
\Return $\mathit{root}$
\end{algorithm2e}

\begin{procedure}[t]
  \caption{BuildDepthFirst($i$, $\mathit{trieNode}$)}
  \label{procedure:bdf}
  \SetKw{KwAND}{and}
  \SetKw{KwOR}{or}
  \SetKw{KwBREAK}{break}
  \SetKwFunction{BDF}{\ref{procedure:bdf}}
  \tcp{add prefix of $r(i)$ to trieNode while right neighbors of $i$ are unique}
  $l \leftarrow 0$; $k \leftarrow i$\;
  \While{
    {$\textsf{true}$}
  }{
    $l \leftarrow l + |\nodelabel(X_k)|$\;
    $\mathsf{visited}[k] \leftarrow \mathsf{true}$\;    
    \tcp{exit loop if right neighbor is possibly non-unique or is visited}
    \lIf{$|X_{r(k)}| < q$ \KwOR $\mathsf{visited}[\lmq(X_{r(k)})] =
      \mathsf{true}$}{
      {\KwBREAK}\;
    }
    $k \leftarrow\lmq(X_{r(k)})$\;
  }
  add new branch from $\mathit{trieNode}$ with string $X_{r(i)}([1:l])$\;\label{algo:prefadd}
  let end of new branch be $\mathit{newTrieNode}$\;
  \tcp{If $|X_{r(k)}| < q$,  there may be multiple right neighbors.}
  \tcp{If $|X_{r(k)}|\geq q$, nothing is done because it has already been visited.}
  \For{$X_c \in \{ X_j \mid (X_k,X_j) \in E_r \}$}{
    \If{$\mathsf{visited}[c] = \mathsf{false}$}{
      \BDF($X_c$, $\mathit{newTrieNode}$)\;
    }
  }
\end{procedure}

\begin{lemma}
  $\trie$ can be constructed in time linear in its size.
\end{lemma}
\begin{proof}
  See Algorithm~\ref{algo:slp2trie}.
  Let $G$ be the $q$-gram neighbor graph.
  We construct $\trie$ in a depth first manner starting at $X_{i_0}$.
  The crux of the algorithm is that rather than 
  computing $\nodelabel()$ separately for each variable,
  we are able to aggregate the $\nodelabel()$s and limit all partial
  decompressions of variables to prefixes of variables, so that
  Lemma~\ref{label:prefix_decompression} can be used.

  Any directed acyclic path on $G$ starting at $X_{i_0}$
  can be segmented into multiple sequences of variables,
  where each sequence $X_{i_j}, \ldots, X_{i_k}$
  is such that $j$ is the only integer in $[j:k]$ such that
  $j = 0 $ or $|X_{r(i_{j-1})}| < q$.
  From Lemma~\ref{lemma:unique_neighbors},
  we have that $X_{i_{j+1}},\ldots,X_{i_k}$
  are uniquely determined.
  If $j>0$, $\nodelabel(X_{i_j})$ is a prefix of $\derive(X_{r(i_j)})$
  since $|X_{r(i_{j-1})}| < q$
  (see Fig.~\ref{fig:qgramneighbor} Right),
  and if $j=0$, $\nodelabel(X_{i_0})$ is again a prefix of
  $\derive(X_{r(i_0)}) = \derive(X_{i_1})$.
  It is not difficult to see that $\nodelabel(X_{i_j})\cdots\nodelabel(X_{i_{k}})$ is
  also a prefix of $X_{r(i_j)}$ since 
  $X_{i_{j+1}},\ldots,X_{i_k}$ are all descendants of $X_{r(i_j)}$, and
  each $\nodelabel()$ extends the
  partially decompressed string to consider consecutive $q$-grams in $X_{r(i_j)}$.
  Since prefixes of variables of SLPs can be decompressed in time
  proportional to the output size with linear time pre-processing
  (Lemma~\ref{label:prefix_decompression}), the lemma follows.
  \qed
\end{proof}

We only illustrate how the character labels are determined
in the pseudo-code of Algorithm~\ref{algo:slp2trie}.
It is straightforward to assign a weight $\VarOcc(X_k)$ to
each node of $\trie$ that corresponds to $\nodelabel(X_k)$.

\begin{lemma}
  \label{lemma:size_of_trie}
  The number of edges in $\trie$ is
  $(q-1) + \sum \{ |t_i|-(q-1) \mid |X_i| \geq q, i=1,\ldots,n\} =
  |T| - \mathit{dup}(q,\mathcal{T})$ where
  \[
  \mathit{dup}(q,\mathcal{T}) 
  = \sum\{ (\VarOcc(X_i) - 1) \cdot (|t_i| - (q-1)) \mid |X_i| \geq q, i=1,\ldots,n \}\}
  \]
\end{lemma}
\begin{proof}
  $(q-1)+\sum \{ |t_i|-(q-1) \mid |X_i| \geq q, i=1,\ldots,n\}$ is straight
  forward from the definition of $\nodelabel(X_i)$ and the
  construction of $\trie$.
  Concerning $\mathit{dup}$, 
  each variable $X_i$ occurs $\VarOcc(X_i)$ times
  in the derivation tree, but only once in the directed spanning tree.  
  This means that for each occurrence after the first,
  the size of $\trie$ is reduced by $|\nodelabel(X_i)|=|t_i| - (q-1)$ compared to $T$. 
  Therefore, the lemma follows.
  \qed
\end{proof}

To efficiently count the weighted $q$-gram frequencies on $\trie$,
we can use suffix trees.
A suffix tree for a trie is defined as a generalized suffix tree
for the set of strings represented in the trie as leaf to root paths.
\footnote{
  When considering leaf to root paths on $\trie$, 
  the direction of the string is the reverse of what is
  in $T$. However, this is merely a matter of representation of the
  output.
}
The following is known.

\begin{lemma}[\cite{shibuya03:_const_suffix_tree_tree_large_alphab}]
  Given a trie of size $m$, 
  the suffix tree for the trie can be constructed in
  $O(m)$ time and space.
\end{lemma}

With a suffix tree, it is a simple exercise to solve the weighted
$q$-gram frequencies problem on $\trie$ in linear time.
In fact, it is known that the suffix array for the common suffix trie
can also be constructed in
linear time~\cite{ferragina09:_compr}, as well as its longest common prefix
array~\cite{kimura11}, which can also be used to solve the problem in
linear time.

\begin{corollary}
  The weighted $q$-gram frequencies problem on a trie of size $m$
  can be solved in $O(m)$ time and space.
\end{corollary}

From the above arguments, the theorem follows.
\begin{theorem}
  The $q$-gram frequencies problem on an SLP $\mathcal{T}$ of size
  $n$, representing string $T$ can be solved in
  $O(\min\{qn,|T| - \mathit{dup}(q,\mathcal{T})\})$ time and space.
\end{theorem}

Note that since each $q\leq |t_i|\leq 2(q-1)$,
and $|\nodelabel(X_i)| = |t_i| - (q-1)$, the
total length of decompressions made by the algorithm, i.e. the size
of the reduced problem, is at least halved and can be as small as $1/q$
(when all $|t_i|=q$, for example, in an SLP that represents LZ78
compression),
compared to the previous $O(qn)$ algorithm.

%% file: experiments.tex
\section{Preliminary Experiments}

We first evaluate the size of the trie $\trie$ induced from the 
right $q$-gram neighbor graph, on which the running time of the new algorithm
of Section~\ref{section:new_algorithm} is dependent.
We used data sets obtained from Pizza \& Chili Corpus,
and constructed SLPs using the RE-PAIR~\cite{LarssonDCC99} compression algorithm.
Each data is of size 200MB.
Table~\ref{table:zsize} shows the sizes of $\trie$ for different values of $q$,
in comparison with the total length of strings $t_i$, 
on which the previous $O(qn)$-time algorithm of Section~\ref{section:qn} works.
We cumulated the lengths of all $t_i$'s only for those satisfying $|t_i| \geq q$,
since no $q$-gram can occur in $t_i$'s with $|t_i| < q$.
Observe that for all values of $q$ and for all data sets,
the size of $\trie$ (i.e., the total number of characters in $\trie$) is 
smaller than those of $t_i$'s and the original string.

\begin{table}[t]
  \caption{
   A comparison of the size of $\trie$ and the total length of strings $t_i$ for 
   SLPs that represent textual data from Pizza \& Chili Corpus. 
   The length of the original text is 209,715,200.
   The SLPs were constructed by RE-PAIR~\cite{LarssonDCC99}.
  }
  \label{table:zsize}
  \begin{center}
    \scriptsize
    \setlength{\tabcolsep}{1pt}
    \renewcommand{\rmdefault}{ptm}
    \renewcommand{\sfdefault}{phv}
    \renewcommand{\ttdefault}{pcr}
    \normalfont
    \input{table_zsize.tex}

  \end{center}
\end{table}

The construction of the suffix tree or array for a trie,
as well as the algorithm for Lemma~\ref{label:prefix_decompression},
require various tools such as level ancestor 
queries~\cite{dietz91:_findin,berkman94:_findin,bender04:_level_ances_probl}
for which we did not have an efficient implementation.
Therefore, we try to assess the practical impact of the reduced problem size
using a simplified version of our new algorithm.
We compared three algorithms ($\nsa$, $\ssa$, $\stsa$) that count the occurrence frequencies of all
$q$-grams in a text given as an SLP.
$\nsa$ is the $O(|T|)$-time algorithm which works on the uncompressed text,
using suffix and LCP arrays.
$\ssa$ is our previous $O(qn)$-time algorithm~\cite{goto11:_fast_minin_slp_compr_strin},
and $\stsa$ is a simplified version of our new algorithm.
$\stsa$ further reduces the weighted $q$-gram frequencies problem on $\trie$,
to a weighted $q$-gram frequencies problem on a single string as follows:
instead of constructing $\trie$, each branch of $\trie$ (on line~\ref{algo:prefadd} of~\ref{procedure:bdf})
is appended into a single string.
The $q$-grams that are represented in the branching edges of $\trie$
can be represented in the single string, by redundantly adding 
$\suffix(X_{r(i)}([1:l]),q-1)$ in front of the string corresponding to
the next branch.
This leads to some duplicate partial decompression, but the resulting
string is still always shorter than the string produced by our
previous algorithm~\cite{goto11:_fast_minin_slp_compr_strin}.
The partial decompression of $X_{r(i)}([1:l])$ is implemented using
a simple $O(h+l)$ algorithm, where $h$ is the height of the SLP
which can be as large as $O(n)$.

All computations were conducted on a Mac Pro (Mid 2010)
with MacOS X Lion 10.7.2,
and 2 x 2.93GHz 6-Core Xeon processors and 64GB Memory,
only utilizing a single process/thread at once.
The program was compiled using the GNU C++ compiler ({\tt g++}) 4.6.2
with the {\tt -Ofast} option for optimization.
The running times were measured in seconds, after reading
the uncompressed text into memory for $\nsa$, and
after reading the SLP that represents the text into memory for $\ssa$ and
$\stsa$.
Each computation was repeated at least 3 times, and the average was taken.

Table~\ref{table:running_time} summarizes the running times of the three algorithms.
$\ssa$ and $\stsa$ computed weighted $q$-gram frequencies on $t_i$ and $\trie$, respectively.
Since the difference between the total length of $t_i$ and
the size of $\trie$ becomes larger as $q$ increases,
$\stsa$ outperforms $\ssa$ when the value of $q$ is not small.
In fact, in Table~\ref{table:running_time} SSA2 was faster than $\ssa$
for all values of $q > 3$.
$\stsa$ was even faster than $\nsa$ on the XML data whenever $q \leq 20$.
What is interesting is that $\stsa$ outperformed $\nsa$ on the ENGLISH data 
when $q = 100$.

\begin{table}[t]
  \caption{
    Running time in seconds for SLPs that represent textual data from Pizza \& Chili Corpus.
    The SLPs were constructed by RE-PAIR~\cite{LarssonDCC99}.
    Bold numbers represent the fastest time for each data and $q$.
    $\stsa$ is faster than $\ssa$ whenever $q>3$.
  }
  \label{table:running_time}
  \begin{center}
    \scriptsize
    \setlength{\tabcolsep}{1pt}
    \renewcommand{\rmdefault}{ptm}
    \renewcommand{\sfdefault}{phv}
    \renewcommand{\ttdefault}{pcr}
    \normalfont
    \input{table_time.tex}

  \end{center}
\end{table}

%% file: table_zsize.tex
\begin{tabular}{|c|r|r|c|r|r|c|r|r|c|r|r|}
 \hline 
 & \multicolumn{2}{c|}{XML} & &  \multicolumn{2}{c|}{DNA} & & \multicolumn{2}{c|}{ENGLISH} & & \multicolumn{2}{c|}{PROTEINS} \\ \hline
% & dblp &  & dna &  & english &  & proteins &  \\ \hline
$q$ & $\sum |t_i|$ & size of $\trie$ & & $\sum |t_i|$ & size of $\trie$ & & $\sum |t_i|$ & size of $\trie$ & & $\sum |t_i|$ & size of $\trie$ \\ \hline
2 & 19,082,988 & 9,541,495 &   & 46,342,894 & 23,171,448 &   & 37,889,802 & 18,944,902 &   & 64,751,926 & 32,375,964\\ \hline
3 & 37,966,315 & 18,889,991 &   & 92,684,656 & 46,341,894 &   & 75,611,002 & 37,728,884 &   & 129,449,835 & 64,698,833\\ \hline
4 & 55,983,397 & 27,443,734 &   & 139,011,475 & 69,497,812 &   & 112,835,471 & 56,066,348 &   & 191,045,216 & 93,940,205\\ \hline
5 & 72,878,965 & 35,108,101 &   & 185,200,662 & 92,516,690 &   & 148,938,576 & 73,434,080 &   & 243,692,809 & 114,655,697\\ \hline
6 & 88,786,480 & 42,095,985 &   & 230,769,162 & 114,916,322 &   & 183,493,406 & 89,491,371 &   & 280,408,504 & 123,786,699\\ \hline
7 & 103,862,589 & 48,533,013 &   & 274,845,524 & 135,829,862 &   & 215,975,218 & 103,840,108 &   & 301,810,933 & 127,510,939\\ \hline
8 & 118,214,023 & 54,500,142 &   & 315,811,932 & 153,659,844 &   & 246,127,485 & 116,339,295 &   & 311,863,817 & 129,618,754\\ \hline
9 & 131,868,777 & 60,045,009 &   & 352,780,338 & 167,598,570 &   & 273,622,444 & 126,884,532 &   & 318,432,611 & 131,240,299\\ \hline
10 & 144,946,389 & 65,201,880 &   & 385,636,192 & 177,808,192 &   & 298,303,942 & 135,549,310 &   & 325,028,658 & 132,658,662\\ \hline
15 & 204,193,702 & 86,915,492 &   & 477,568,585 & 196,448,347 &   & 379,441,314 & 157,558,436 &   & 347,993,213 & 138,182,717\\ \hline
20 & 255,371,699 & 104,476,074 &   & 497,607,690 & 200,561,823 &   & 409,295,884 & 162,738,812 &   & 364,230,234 & 142,213,239\\ \hline
50 & 424,505,759 & 157,069,100 &   & 530,329,749 & 206,796,322 &   & 429,380,290 & 165,882,006 &   & 416,966,397 & 156,257,977\\ \hline
100 & 537,677,786 & 192,816,929 &   & 536,349,226 & 207,838,417 &   & 435,843,895 & 167,313,028 &   & 463,766,667 & 168,544,608\\ \hline

\end{tabular}

%% file: table_time.tex
\begin{tabular}{|c|r|r|r|c|r|r|r|c|r|r|r|c|r|r|r|}
 \hline 
 & \multicolumn{3}{c|}{XML} & &  \multicolumn{3}{c|}{DNA} & &  \multicolumn{3}{c|}{ENGLISH} & & \multicolumn{3}{c|}{PROTEINS} \\ \hline
$q$ & $\nsa$ & $\ssa$ & $\stsa$ & & $\nsa$ & $\ssa$ & $\stsa$ & & $\nsa$ & $\ssa$ & $\stsa$ & & $\nsa$ & $\ssa$ & $\stsa$ \\ \hline
2 & 41.67 & \textbf{6.53} & 7.63 &   & 61.28 & \textbf{19.27} & 22.73 &   & 56.77 & \textbf{16.31} & 19.23 &   & 60.16 & \textbf{27.13} & 30.71 \\ \hline
3 & 41.46 & 10.96 & \textbf{10.92} &   & 61.28 & \textbf{29.14} & 31.07 &   & 56.77 & 25.58 & \textbf{25.57} &   & 60.53 & \textbf{47.53} & 50.65 \\ \hline
4 & 41.87 & 16.27 & \textbf{14.5} &   & 61.65 & 42.22 & \textbf{41.69} &   & 56.77 & 37.48 & \textbf{34.95} &   & \textbf{60.86} & 74.89 & 73.51 \\ \hline
5 & 41.85 & 21.33 & \textbf{17.42} &   & 61.57 & 56.26 & \textbf{54.21} &   & 57.09 & 49.83 & \textbf{45.21} &   & \textbf{60.53} & 101.64 & 79.1 \\ \hline
6 & 41.9 & 25.77 & \textbf{20.07} &   & \textbf{60.91} & 73.11 & 68.63 &   & 57.11 & 62.91 & \textbf{55.28} &   & \textbf{61.18} & 123.74 & 75.83 \\ \hline
7 & 41.73 & 30.14 & \textbf{21.94} &   & \textbf{60.89} & 90.88 & 82.85 &   & \textbf{56.64} & 75.69 & 63.35 &   & \textbf{61.14} & 136.12 & 72.62 \\ \hline
8 & 41.92 & 34.22 & \textbf{23.97} &   & \textbf{61.57} & 110.3 & 93.46 &   & \textbf{57.27} & 87.9 & 69.7 &   & \textbf{61.39} & 142.29 & 71.08 \\ \hline
9 & 41.92 & 37.9 & \textbf{25.08} &   & \textbf{61.26} & 127.29 & 96.07 &   & \textbf{57.09} & 100.24 & 73.63 &   & \textbf{61.36} & 148.12 & 69.88 \\ \hline
10 & 41.76 & 41.28 & \textbf{26.45} &   & \textbf{60.94} & 143.31 & 96.26 &   & \textbf{57.43} & 110.85 & 75.68 &   & \textbf{61.42} & 149.73 & 69.34 \\ \hline
15 & 41.95 & 58.21 & \textbf{32.21} &   & \textbf{61.72} & 190.88 & 84.86 &   & \textbf{57.31} & 146.89 & 70.63 &   & \textbf{60.42} & 160.58 & 66.57 \\ \hline
20 & 41.82 & 74.61 & \textbf{39.62} &   & \textbf{61.36} & 203.03 & 83.13 &   & \textbf{57.65} & 161.12 & 64.8 &   & \textbf{61.01} & 165.03 & 66.09 \\ \hline
50 & \textbf{42.07} & 134.38 & 53.98 &   & \textbf{61.73} & 216.6 & 78.0 &   & \textbf{57.02} & 166.67 & 57.89 &   & \textbf{61.05} & 181.14 & 66.36 \\ \hline
100 & \textbf{41.81} & 181.23 & 60.18 &   & \textbf{61.46} & 217.05 & 75.91 &   & 57.3 & 166.67 & \textbf{56.86} &   & \textbf{60.69} & 197.33 & 69.9 \\ \hline

\end{tabular}